\documentclass[letterpaper,10pt,conference]{ieeeconf}

\usepackage[T1]{fontenc}
\usepackage{amsmath,amssymb,amsfonts,bm}

\usepackage{amsthm}
\usepackage{algorithm,algorithmic}
\usepackage{booktabs}
\usepackage{graphicx}
\usepackage{tabularx}
\usepackage{xcolor}
\usepackage{url}
\usepackage{stfloats}
\usepackage{flushend}

\IEEEoverridecommandlockouts
\newtheorem{thm}{Theorem}
\newtheorem{lemma}{Lemma}
\newtheorem{proposition}{Proposition}
\newtheorem{corollary}{Corollary}
\theoremstyle{definition}

\newtheorem{assumption}{Assumption}

\newcommand{\R}{\mathbb{R}}
\newcommand{\M}{\mathcal{M}}
\newcommand{\X}{\mathbb{X}}
\newcommand{\Pone}{\mathcal{P}_1}
\newcommand{\cW}{\mathcal{W}_{\bm v,B}}
\newcommand{\cA}{\mathcal{A}}
\newcommand{\cL}{\mathcal{L}}
\newcommand{\dd}{\,\mathrm{d}}
\newcommand{\ind}{\mathbf{1}}
\newcommand{\eps}{\varepsilon}

\title{\LARGE\bf
Mode-Weighted Transport Certificates for State-Dependent Reflected Switching Diffusions}

\author{Yutong Zhu and Ye Zhang%
\thanks{The authors are with the School of Astronautics, Northwestern Polytechnical University, Xi'an 710072, China.}}

\begin{document}
\maketitle
\thispagestyle{empty}
\pagestyle{empty}

\begin{abstract}
State-dependent switching diffusions can contract in law even when individual modes are expansive. We establish a computable, constant-one contraction certificate on $\R^n$ and compact convex domains with normal reflection. The mode-weighted transport cost combines spatial separation with a discrete graph penalty and retains spatial dependence when the modes differ. A synchronous Brownian coupling and a maximal coupling of jump clocks in equal modes lead to separate equal- and unequal-mode generator inequalities. We give a complete regularization argument for the nonsmooth cost, including the drift and jump errors, the sign of reflection, and the limiting It\^o-Dynkin calculation. A resolvent proof transfers contraction to probability-valued weak solutions of the no-flux Fokker-Planck-Kolmogorov system, while an explicit Foster-Lyapunov condition ensures existence and uniqueness of an invariant law on the whole space. For fixed weight ordering and decay rate, certificate synthesis is a semi-infinite linear feasibility problem; on compact domains, a Lipschitz-buffered finite mesh provides a rigorous test. One-dimensional and planar examples illustrate the certificate and its numerical verification.
\end{abstract}

\section{Introduction}

Switching diffusions couple a continuous stochastic state to a finite-state mode. When the transition rates depend on position, the discrete environment and the continuous dynamics form a feedback loop. Well-posedness and Feller properties for such processes are developed in \cite{nguyen2017properties,nguyen2025hybrid}, while Lyapunov, matrix, and recurrence criteria show that switching can stabilize a system even if some modes are not stable in isolation \cite{basak1999stability,wu2013stability,shao2014stability}. The state dependence that makes this mechanism useful also complicates a direct stability test: the rate matrix changes along the continuous trajectory, and two copies at different positions no longer share the same jump clocks.

Coupling and transport methods compare probability laws directly \cite{villani2009optimal,chen2021optimal}. Wasserstein contraction on closed convex domains has been studied for reflected diffusions \cite{lekang2021wasserstein}, though without the state-dependent switching and mode-dependent cost considered here. For switching processes, exponential convergence has been obtained from Wasserstein contraction of the constituent dynamics and weak Harris arguments \cite{cloez2015exponential}, from $M$-matrix and Perron-Frobenius criteria \cite{shao2015ergodicity}, and from contraction-on-average combined with a Lyapunov function \cite{tong2016moment}. State-dependent switching has also been treated through successful coupling and numerical approximation \cite{shao2018invariant}, comparison of transition mechanisms \cite{shao2024comparison}, and piecewise-constant approximation of the rate matrix \cite{shao2023ergodicity}, and Euler-Maruyama approximation \cite{jin2025euler}. Related Wasserstein results cover subgeometric convergence \cite{lazic2022subgeometric} and functional switching diffusions with memory \cite{shi2022ergodicity}.

Our objective is a global pairwise contraction estimate,
\begin{equation*}
 \mathcal W(\mu P_t,\widetilde\mu P_t)
 \le e^{-\eta t}\mathcal W(\mu,\widetilde\mu),\qquad t\ge0,
\end{equation*}
for every pair of initial laws, with a unit prefactor and a computable rate $\eta$. This objective differs from convergence to a previously established invariant law. It also requires a cost that remains informative before two coupled modes meet. For example, the hybrid distance in \cite{cloez2015exponential} assigns a fixed cost to unequal modes and uses spatial separation only when the modes agree. The mode-weighted cost introduced below instead retains a continuous separation term even when the modes differ. This term is necessary in our unequal-mode generator bound.

On a bounded domain, the boundary condition must be part of the stochastic argument. A no-flux Fokker-Planck-Kolmogorov (FPK) equation is represented by a reflected diffusion, not by an unconstrained It\^o equation. Reflected diffusions and their Skorokhod formulations are classical \cite{tanaka1979stochastic,lions1984stochastic}, and reflected regime-switching examples already occur in Wasserstein analysis \cite{shao2015ergodicity}. For a pairwise estimate, the finite-variation reflection terms must be retained until their signs are determined. Convexity makes the normal-reflection contribution to the Euclidean separation nonpositive.

The resulting certificate has two parts. Equal-mode inequalities combine one-sided drift contraction, simultaneous jumps, and the unmatched-clock penalty caused by rate sensitivity. Unequal-mode inequalities account for drift mismatch and for every jump of either component. Synchronous Brownian driving and maximally coupled equal-target jump clocks lead to a pairwise semigroup estimate once the nonsmooth transport cost is regularized and the resulting error is controlled. A resolvent argument transfers it to weak measure solutions of the no-flux FPK system. Once the ordering of the mode weights and $\eta$ are fixed, all certificate inequalities are affine in the weights and graph costs. Constraint generation therefore applies, and an explicit mesh buffer distinguishes a numerical screen from a proof over the full domain.

\section{Reflected Switching Diffusion and the FPK System}
\label{sec:model}

\subsection{Hybrid state process}

Let $\M=\{1,\ldots,M\}$ and let $\Omega$ be either $\R^n$ or a compact convex domain with $C^2$ boundary. On a reflected domain, let $\boldsymbol n(x)$ be the outward unit normal at $x\in\partial\Omega$. The hybrid state space is $\X=\Omega\times\M$. Let $(\mathcal F_t)_{t\ge0}$ be a filtration satisfying the usual conditions, let $W_t$ be an $n$-dimensional standard $(\mathcal F_t)$-Brownian motion, and let $\nu\ge0$ be the constant diffusivity.

The continuous component satisfies the normally reflected equation
\begin{equation}
 \dd X_t=f_{\sigma_t}(X_t)\dd t+\sqrt{2\nu}\,\dd W_t
 -\boldsymbol n(X_t)\dd K_t, 
 \label{eq:reflected_sde}
\end{equation}
and the mode has conditional transition rates
\begin{align}
 &\Pr(\sigma_{t+\Delta}=j\mid\mathcal F_t,X_t=x,\sigma_t=i)\notag\\
 &\hspace{3em}=\lambda_{ij}(x)\Delta+o(\Delta),\quad j\ne i.
 \label{eq:switch_rates}
\end{align}
Here $K_t$ is a continuous, nondecreasing reflection regulator, supported on the boundary:
\begin{equation}
 K_0=0,\qquad \int_0^t\ind_{\{X_s\in\Omega^\circ\}}\dd K_s=0.
 \label{eq:local_time_support}
\end{equation}
For $\Omega=\R^n$, the reflection term is absent and $K_t\equiv0$. The switching does not reset $X_t$. We set $\lambda_{ii}(x)=-\sum_{j\ne i}\lambda_{ij}(x)$.

Figure \ref{fig:hybrid_geometry} separates the two objects that enter the analysis. The arrows $\lambda_{ij}(x)$ describe physical mode switches of the Markov process, whereas $c_{\bm v,B}((x,i),(y,j))$ is the comparison cost assigned to a pair of hybrid states. In particular, $\beta_{ij}$ is a transport penalty, not a transition rate.

\begin{figure}[htb!]
\centering
\includegraphics[width=0.92\linewidth]{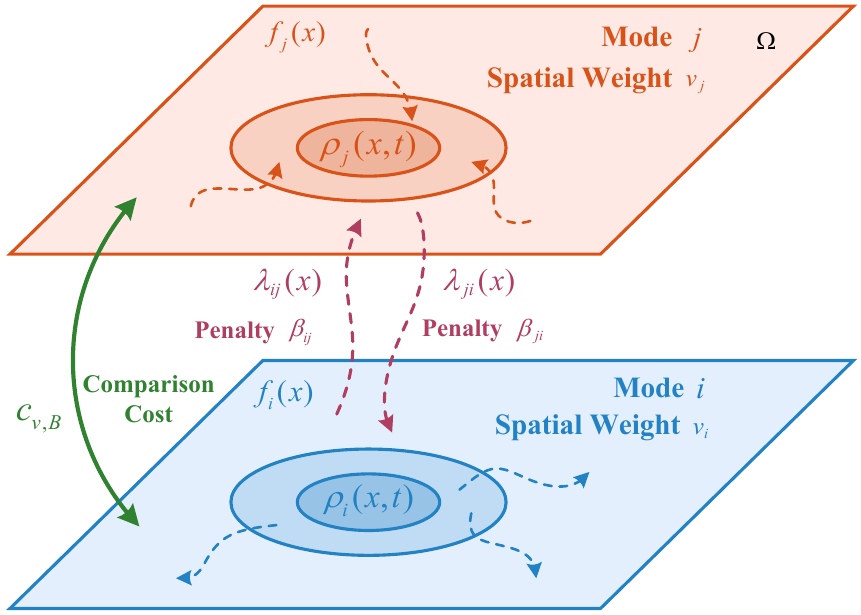}
\caption{Hybrid geometry of the model and certificate. Each mode supports a reflected continuous probability flow on $\Omega$ and the state-dependent clocks connect the mode layers. The comparison cost retains both the weighted spatial separation and the discrete penalty when the two copies occupy different modes.}
\label{fig:hybrid_geometry}
\end{figure}

\begin{assumption}[Coefficient regularity and generator core]
\label{ass:wellposed}
Each $f_i$ is globally Lipschitz on $\Omega$ and has at most linear growth when $\Omega=\R^n$. Each $\lambda_{ij}$, $i\ne j$, is nonnegative, bounded, and Lipschitz. The Neumann test-function class introduced below is a core for the generator of the Feller semigroup on $C_0(\X)$.
\end{assumption}

The reflected diffusion is well posed under the stated geometric and regularity assumptions \cite{tanaka1979stochastic,lions1984stochastic}. Interlacing it with the bounded state-dependent jump clocks gives a conservative, nonexplosive strong Markov process. The corresponding whole-space construction and its Feller properties are studied in \cite{nguyen2017properties}. The core property used below is part of Assumption~\ref{ass:wellposed}; it is not inferred solely from coefficient regularity.

\subsection{Generator, no-flux equation, and weak solutions}

We use the following explicit test-function classes:
\begin{align}
\mathcal D_N&=C_c^2(\R^n;\R^M),&&\Omega=\R^n,\notag\\
\mathcal D_N&=\{\varphi\in C^2(\overline\Omega;\R^M):\notag\\
&\hspace{1em}\nabla\varphi_i(x)^\top\boldsymbol n(x)=0,\notag\\
&\hspace{1em}x\in\partial\Omega,\ i\in\M\},&&\Omega\text{ compact}.
\label{eq:test_core}
\end{align}
Assumption~\ref{ass:wellposed} requires the corresponding class to be a core for the generator $\cA$ of the Feller semigroup $(P_t)$ on $C_0(\X)$. For $\varphi\in\mathcal D_N$, the generator is
\begin{align}
 (\cA\varphi)_i(x)
 &=f_i(x)^\top\nabla\varphi_i(x)+\nu\Delta\varphi_i(x)\notag\\
 &\quad+\sum_{j\ne i}\lambda_{ij}(x)
 [\varphi_j(x)-\varphi_i(x)].
 \label{eq:generator}
\end{align}
When the mode marginals have sufficiently regular densities and boundary traces, the corresponding Fokker-Planck-Kolmogorov (FPK) system is
\begin{equation}
 \partial_t\rho_i=-\nabla\!\cdot(f_i\rho_i)+\nu\Delta\rho_i
 +\sum_{j=1}^M\lambda_{ji}\rho_j,\qquad i\in\M,
 \label{eq:fpk}
\end{equation}
with no-flux boundary condition
\begin{equation}
 \big(f_i\rho_i-\nu\nabla\rho_i\big)^\top\boldsymbol n=0
 \quad\text{on }\partial\Omega.
 \label{eq:no_flux}
\end{equation}
The Neumann condition on test functions eliminates their reflection boundary contribution. For smooth densities, the corresponding adjoint condition is the no-flux condition~\eqref{eq:no_flux}. The weak formulation below does not require densities to exist.

A narrowly continuous probability-valued curve $\mu_t=(\mu_{1,t},\ldots,\mu_{M,t})$ on $\X$ is a weak solution if, for every $\varphi\in\mathcal D_N$,
\begin{align}
 \sum_i\int_\Omega\varphi_i\dd\mu_{i,t}
 &=\sum_i\int_\Omega\varphi_i\dd\mu_{i,0}\notag\\
 &\quad+\int_0^t\sum_i\int_\Omega
 (\cA\varphi)_i\dd\mu_{i,s}\dd s.
 \label{eq:weak_fpk}
\end{align}

\begin{proposition}[Semigroup-FPK correspondence]
\label{prop:semigroup_fpk}
Under Assumption \ref{ass:wellposed}, $\mu_t=\mu_0P_t$ satisfies \eqref{eq:weak_fpk}. Conversely, every narrowly continuous probability-valued weak solution of \eqref{eq:weak_fpk} equals $\mu_0P_t$.
\end{proposition}

\begin{proof}
For $\varphi\in\mathcal D_N$, Dynkin's formula for a localized process, followed by nonexplosion and boundedness of $\varphi$ and $\cA\varphi$, gives
\begin{equation*}
 P_t\varphi-\varphi=\int_0^tP_s\cA\varphi\dd s.
\end{equation*}
Integration against $\mu_0$ proves that $\mu_0P_t$ satisfies \eqref{eq:weak_fpk}. It remains to show that the weak identity does not admit a second probability-valued evolution.

Fix $\alpha>0$ and define the finite Laplace-transformed measure
\begin{equation*}
 \overline\mu_\alpha(\varphi)
 =\int_0^\infty e^{-\alpha t}\mu_t(\varphi)\dd t,
 \qquad \varphi\in C_0(\X).
\end{equation*}
For $\varphi\in\mathcal D_N$, the probability-valued property gives
\begin{equation*}
 |\mu_s(\cA\varphi)|\le\|\cA\varphi\|_\infty.
\end{equation*}
Consequently,
\begin{align*}
 \int_0^\infty e^{-\alpha t}\int_0^t
 |\mu_s(\cA\varphi)|\dd s\dd t
 &\le\|\cA\varphi\|_\infty
 \int_0^\infty t e^{-\alpha t}\dd t\\
 &=\frac{\|\cA\varphi\|_\infty}{\alpha^2}<\infty.
\end{align*}
This justifies Fubini's theorem when substituting \eqref{eq:weak_fpk} into the Laplace transform and gives
\begin{align*}
 \overline\mu_\alpha(\varphi)
 &=\frac{1}{\alpha}\mu_0(\varphi)
   +\int_0^\infty\!\int_s^\infty
     e^{-\alpha t}\dd t\,\mu_s(\cA\varphi)\dd s\\
 &=\frac{1}{\alpha}\mu_0(\varphi)
   +\frac{1}{\alpha}\overline\mu_\alpha(\cA\varphi).
\end{align*}
Consequently,
\begin{equation}
 \overline\mu_\alpha\big((\alpha I-\cA)\varphi\big)
 =\mu_0(\varphi),\qquad \varphi\in\mathcal D_N.
 \label{eq:resolvent_identity}
\end{equation}
For comparison, define the finite measure
\begin{equation*}
 \overline\mu_\alpha^P
 =\int_0^\infty e^{-\alpha t}\mu_0P_t\dd t.
\end{equation*}
The semigroup resolvent identity yields
\begin{equation}
 \overline\mu_\alpha^P((\alpha I-\cA)\varphi)
 =\mu_0(\varphi),\quad\varphi\in\mathcal D_N.
 \label{eq:semigroup_resolvent}
\end{equation}
Thus $\overline\mu_\alpha$ and $\overline\mu_\alpha^P$ agree on $(\alpha I-\cA)\mathcal D_N$. Since $(P_t)$ is a strongly continuous contraction semigroup on $C_0(\X)$, the resolvent has range $C_0(\X)$. The core assumption implies that $(\alpha I-\cA)\mathcal D_N$ is dense in $C_0(\X)$. Both finite measures define bounded linear functionals on $C_0(\X)$, so their agreement extends to all of $C_0(\X)$. For every $g\in C_0(\X)$, the scalar functions $t\mapsto\mu_t(g)$ and $t\mapsto\mu_0P_t(g)$ therefore have identical Laplace transforms for every $\alpha>0$. The former is continuous by narrow continuity, and the latter by the Feller property. Uniqueness of scalar Laplace transforms gives $\mu_t(g)=\mu_0P_t(g)$ for every $t\ge0$. Since $C_0(\X)$ separates finite Borel measures, $\mu_t=\mu_0P_t$ \cite{ethier1986markov}.
\end{proof}

Thus the weak solution class in \eqref{eq:weak_fpk} coincides with the semigroup evolution. On bounded domains, this correspondence is based on the reflected process and its Neumann generator.

\section{Mode-Weighted Hybrid Transport Discrepancy}
\label{sec:transport}

Let $\bm v=(v_1,\ldots,v_M)^\top\in\R_{>0}^M$. Let $B=[\beta_{ij}]$ be symmetric, with $\beta_{ii}=0$, $\beta_{ij}>0$ for $i\ne j$, and
\begin{equation}
 \beta_{ij}\le\beta_{ik}+\beta_{kj},\qquad i,j,k\in\M.
 \label{eq:graph_metric}
\end{equation}
Define $a_{ii}=v_i$ and $a_{ij}=\min\{v_i,v_j\}$ for $i\ne j$. For $z=(x,i)$ and $\tilde z=(y,j)$, set
\begin{equation}
 c_{\bm v,B}(z,\tilde z)
 =\inf_{q\in\Omega}\big[v_i\|x-q\|+\beta_{ij}+v_j\|q-y\|\big].
 \label{eq:path_cost}
\end{equation}

\begin{lemma}[Closed form of the cost]
\label{lem:cost_closed}
For all $x,y\in\Omega$ and $i,j\in\M$,
\begin{equation}
 c_{\bm v,B}((x,i),(y,j))=a_{ij}\|x-y\|+\beta_{ij}.
 \label{eq:closed_cost}
\end{equation}
\end{lemma}

\begin{proof}
Put $r=\|x-y\|$. Since $v_i\ge a_{ij}$ and $v_j\ge a_{ij}$, the triangle inequality gives, for every $q\in\Omega$,
\begin{align*}
 v_i\|x-q\|+v_j\|q-y\|
 &\ge a_{ij}\big(\|x-q\|+\|q-y\|\big)\\
 &\ge a_{ij}r.
\end{align*}
This proves the lower bound in \eqref{eq:closed_cost}. If $v_i\le v_j$, the admissible choice $q=y$ makes the path cost $v_ir+\beta_{ij}=a_{ij}r+\beta_{ij}$. If $v_j\le v_i$, the choice $q=x$ gives the same value with $v_j=a_{ij}$. The lower bound is therefore attained in both cases.
\end{proof}

For each mode pair $(i,j)$, define the \emph{mode-pair cost}
\begin{equation}
 D^{ij}(x,y):=c_{\bm v,B}((x,i),(y,j))
 =a_{ij}\|x-y\|+\beta_{ij}.
 \label{eq:mode_pair_cost}
\end{equation}
In particular, $D^{ii}(x,y)=v_i\|x-y\|$, because $a_{ii}=v_i$ and $\beta_{ii}=0$.

Let $\mathcal P_1(\X)$ be the set of Borel probability measures on $\X$ with finite first spatial moment,
\begin{equation*}
 \int_\X\|x\|\,\mu(\dd(x,i))<\infty.
\end{equation*}
For $\mu,\tilde\mu\in\Pone(\X)$, let $\Pi(\mu,\tilde\mu)$ be the set of joint probability measures on $\X\times\X$ with marginals $\mu$ and $\tilde\mu$. Define
\begin{equation}
 \cW(\mu,\tilde\mu)
 =\inf_{\pi\in\Pi(\mu,\tilde\mu)}
 \int_{\X\times\X}c_{\bm v,B}(z,\tilde z)\pi(\dd z,\dd\tilde z).
 \label{eq:transport_discrepancy}
\end{equation}
The resulting functional is symmetric and nonnegative; the separation property follows from the bounds proved below. With unequal $v_i$, however, \eqref{eq:closed_cost} need not satisfy the triangle inequality. For example, take $v_1=2$, $v_2=1$, $\beta_{12}=b$, and suppose that $0,R,2R\in\Omega$ with $R>b$. For $z_0=(0,1)$, $z_1=(R,2)$, and $z_2=(2R,1)$,
\begin{equation*}
 c_{\bm v,B}(z_0,z_2)=4R
 >2R+2b=c_{\bm v,B}(z_0,z_1)+c_{\bm v,B}(z_1,z_2).
\end{equation*}
The same counterexample applied to Dirac measures rules out a metric on probability laws. We therefore call \eqref{eq:transport_discrepancy} a Wasserstein-type discrepancy rather than a metric. If all $v_i$ coincide, it is the ordinary $1$-Wasserstein distance induced by the product metric $v\|x-y\|+\beta_{ij}$.

Although the triangle inequality may fail, the discrepancy has the same convergence topology as a standard product-space Wasserstein distance. Let $d_0((x,i),(y,j))=\|x-y\|+\ind_{\{i\ne j\}}$ and denote its $1$-Wasserstein distance by $W_{d_0}$. For $M\ge2$, define
\begin{align*}
 \underline m&=\min\{\min_i v_i,\min_{i\ne j}\beta_{ij}\},\\
 \overline m&=\max\{\max_i v_i,\max_{i\ne j}\beta_{ij}\}.
\end{align*}
When $M=1$, use $\underline m=\overline m=v_1$. In either case, the pointwise bounds
\begin{equation}
 \underline m\,d_0(z,\tilde z)\le c_{\bm v,B}(z,\tilde z)
 \le\overline m\,d_0(z,\tilde z)
 \label{eq:cost_equivalence}
\end{equation}
imply
\begin{equation}
 \underline m W_{d_0}(\mu,\tilde\mu)
 \le\cW(\mu,\tilde\mu)
 \le\overline m W_{d_0}(\mu,\tilde\mu).
 \label{eq:transport_equivalence}
\end{equation}
Thus $\cW$ controls both the continuous first moment and the mismatch probability of the discrete modes. For completeness, the bounds also prove the separation assertion above. Indeed, if $i=j$, then $c_{\bm v,B}=v_i\|x-y\|$ and if $i\ne j$, both $a_{ij}\ge\underline m$ and $\beta_{ij}\ge\underline m$. This gives the left inequality in \eqref{eq:cost_equivalence} and the right one is obtained in the same two cases from $a_{ij},v_i,\beta_{ij}\le\overline m$. Integrating against an arbitrary plan and taking infima proves \eqref{eq:transport_equivalence}. Since $d_0$ is a complete product metric and $W_{d_0}$ separates probability measures in $\Pone(\X)$ \cite{villani2009optimal}, $\cW(\mu,\tilde\mu)=0$ implies $\mu=\tilde\mu$.

\begin{lemma}[Coupling reduction]
\label{lem:coupling_reduction}
Suppose there is a measurable Markov coupling kernel $Q_t((z,\tilde z),\cdot)$ of $P_t(z,\cdot)$ and $P_t(\tilde z,\cdot)$ satisfying
\begin{equation}
 \int c_{\bm v,B}(z',\tilde z')
 Q_t((z,\tilde z),\dd z'\dd\tilde z')
 \le e^{-\eta t}c_{\bm v,B}(z,\tilde z).
 \label{eq:point_coupling}
\end{equation}
Then $\cW(\mu P_t,\tilde\mu P_t)\le e^{-\eta t}\cW(\mu,\tilde\mu)$.
\end{lemma}

\begin{proof}
Fix $\pi\in\Pi(\mu,\tilde\mu)$ and define
\begin{equation*}
 \Gamma_t(A)=\int_{\X\times\X}Q_t((z,\tilde z),A)
 \pi(\dd z\dd\tilde z).
\end{equation*}
The first marginal of $\Gamma_t$ is $\mu P_t$ and the second is $\tilde\mu P_t$, so $\Gamma_t$ is an admissible transport plan. Tonelli's theorem and \eqref{eq:point_coupling} yield
\begin{align*}
 \cW(\mu P_t,\tilde\mu P_t)
 &\le\int c_{\bm v,B}\dd\Gamma_t\\
 &\le e^{-\eta t}\int c_{\bm v,B}\dd\pi.
\end{align*}
Taking the infimum over $\pi$ proves the claim.  The construction composes an initial transport plan with a Markov coupling kernel \cite{villani2009optimal}; no optimal initial plan is required.
\end{proof}

\section{Generator Certificate for Exponential Contraction}
\label{sec:certificate}

\subsection{Coefficient envelopes}

\begin{assumption}[Equal-mode one-sided contraction]
\label{ass:drift}
For each $i\in\M$, there is $c_i\in\R$ such that
\begin{equation}
 \langle x-y,f_i(x)-f_i(y)\rangle\le-c_i\|x-y\|^2.
 \label{eq:same_drift}
\end{equation}
Negative $c_i$ are allowed and represent locally expansive modes.
\end{assumption}

\begin{assumption}[Rate bounds]
\label{ass:rates}
For $i\ne k$, define
\begin{equation}
 \underline\lambda_{ik}=\inf_{x\in\Omega}\lambda_{ik}(x),\quad
 \overline\lambda_{ik}=\sup_{x\in\Omega}\lambda_{ik}(x),
 \label{eq:rate_bounds}
\end{equation}
and assume $|\lambda_{ik}(x)-\lambda_{ik}(y)|\le L_{ik}\|x-y\|$. For $s\in\R$, write $s^+=\max\{s,0\}$ and $s^-=\min\{s,0\}$.
\end{assumption}

\begin{assumption}[Unequal-mode drift envelope]
\label{ass:cross_drift}
For every $i\ne j$, constants $\kappa_{ij}\in\R$ and $h_{ij}\ge0$ satisfy
\begin{equation}
 \frac{\langle x-y,f_i(x)-f_j(y)\rangle}{\|x-y\|}
 \le\kappa_{ij}\|x-y\|+h_{ij},\qquad x\ne y.
 \label{eq:cross_envelope}
\end{equation}
When $x=y$, we require
\begin{equation}
 \|f_i(x)-f_j(x)\|\le h_{ij},\qquad x\in\Omega,\quad i\ne j,
 \label{eq:cross_envelope_diagonal}
\end{equation}
This condition explicitly bounds the drift mismatch when the continuous states coincide. It also supplies the bound required at $R=0$ in the regularized generator calculation below.
\end{assumption}

The $h_{ij}$ term measures the mismatch of the two vector fields near coincident positions, whereas $\kappa_{ij}$ controls its growth with separation.  To see that finite envelopes exist on a compact domain, let $L_i^f$ be a Lipschitz constant of $f_i$ and set $h_{ij}=\sup_{z\in\Omega}\|f_i(z)-f_j(z)\|$. The decomposition
\begin{equation*}
 f_i(x)-f_j(y)=[f_i(x)-f_i(y)]+[f_i(y)-f_j(y)]
\end{equation*}
and Cauchy-Schwarz give \eqref{eq:cross_envelope} with $\kappa_{ij}=L_i^f$.  Sharper one-sided estimates may make $\kappa_{ij}$ smaller and loose global choices can make the certificate infeasible.

Define the equal-mode margin
\begin{align}
 S_i(\bm v,B)=&-c_iv_i+\sum_{k\ne i}\big[(v_k-v_i)^+\overline\lambda_{ik}
 +(v_k-v_i)^-\underline\lambda_{ik}\big]\nonumber\\
 &+\sum_{k\ne i}L_{ik}\beta_{ik},
 \label{eq:same_margin}
\end{align}
and, for $i\ne j$, $r=\|x-y\|$,
\begin{align}
 C_{ij}(x,y;\bm v,B)=&a_{ij}(\kappa_{ij}r+h_{ij})\nonumber\\
 &+\sum_{k\ne i}\lambda_{ik}(x)
 [(a_{kj}-a_{ij})r+\beta_{kj}-\beta_{ij}]\nonumber\\
 &+\sum_{k\ne j}\lambda_{jk}(y)
 [(a_{ik}-a_{ij})r+\beta_{ik}-\beta_{ij}].
 \label{eq:cross_margin}
\end{align}

\subsection{Reflection geometry}

\begin{lemma}[Normal reflection is nonexpansive]
\label{lem:reflection}
Let $x,y\in\Omega$. If $x\in\partial\Omega$, then $\langle x-y,\boldsymbol n(x)\rangle\ge0$, while if $y\in\partial\Omega$, then $\langle x-y,\boldsymbol n(y)\rangle\le0$. Consequently, normal reflection contributes nonpositively to the upper Dini differential of the separation of two synchronously driven solutions of \eqref{eq:reflected_sde}.
\end{lemma}

\begin{proof}
At a boundary point $x$, convexity places the whole set in the supporting half-space
\begin{equation*}
 \langle z-x,\boldsymbol n(x)\rangle\le0,
 \qquad z\in\Omega.
\end{equation*}
Taking $z=y$ gives $\langle x-y,\boldsymbol n(x)\rangle\ge0$. Applying the same argument at $y$ with $z=x$ gives $\langle x-y,\boldsymbol n(y)\rangle\le0$.

Now drive the two reflected equations synchronously and put $\Delta_t=X_t-\tilde X_t$ and $r_\eps(\Delta)=(\|\Delta\|^2+\eps^2)^{1/2}$. The Brownian terms cancel in $\Delta_t$. The ordinary chain rule for the remaining continuous finite-variation part shows that the two reflection terms in $\dd r_\eps(\Delta_t)$ are
\begin{align*}
 &-\frac{\langle\Delta_t,\boldsymbol n(X_t)\rangle}
 {r_\eps(\Delta_t)}\dd K_t\\
 &\quad+\frac{\langle\Delta_t,\boldsymbol n(\tilde X_t)\rangle}
 {r_\eps(\Delta_t)}\dd\tilde K_t.
\end{align*}
By \eqref{eq:local_time_support}, the first measure is carried by $X_t\in\partial\Omega$ and the second by $\tilde X_t\in\partial\Omega$. The two supporting-half-space inequalities therefore make each displayed term nonpositive. These signs hold for every $\eps>0$, including at $\Delta_t=0$. In the limit $\eps\downarrow0$ they give the stated nonpositive upper-Dini contribution. The smooth cost used in the full It\^o-Dynkin argument is defined explicitly in the proof of Theorem~\ref{thm:contraction} \cite{tanaka1979stochastic,lions1984stochastic}.
\end{proof}

\subsection{Main result}

\begin{thm}[Pairwise contraction]
\label{thm:contraction}
Suppose Assumptions \ref{ass:wellposed}-\ref{ass:cross_drift} hold. If there exist $\bm v>0$, a graph cost $B$, and $\eta>0$ such that
\begin{align}
 S_i(\bm v,B)&\le-\eta v_i, &&i\in\M, \label{eq:certificate_same}\\
 C_{ij}(x,y;\bm v,B)&\le-\eta(a_{ij}\|x-y\|+\beta_{ij}),
 &&i\ne j,\ x,y\in\Omega, \label{eq:certificate_cross}
\end{align}
then, for all $\mu_0,\tilde\mu_0\in\Pone(\X)$,
\begin{equation}
 \cW(\mu_0P_t,\tilde\mu_0P_t)
 \le e^{-\eta t}\cW(\mu_0,\tilde\mu_0).
 \label{eq:semigroup_contraction}
\end{equation}
By Proposition~\ref{prop:semigroup_fpk}, the same estimate holds for any two probability-valued weak solutions of \eqref{eq:weak_fpk} with initial laws in $\Pone(\X)$.
\end{thm}

\begin{proof}
We construct a Markov coupling and first identify the bounds for the original cost at positive separation. We then prove the required generator inequality for a smooth approximation, which justifies the argument at coincident states and the subsequent It\^o-Dynkin step.

\emph{Step 1: Markovian coupling and its marginals.}
On a state $(x,i,y,j)\in\X^2$, drive both continuous components by the same Brownian motion:
\begin{align*}
 \dd X_t&=f_i(X_t)\dd t+\sqrt{2\nu}\dd W_t
 -\boldsymbol n(X_t)\dd K_t,\\
 \dd\tilde X_t&=f_j(\tilde X_t)\dd t+\sqrt{2\nu}\dd W_t
 -\boldsymbol n(\tilde X_t)\dd\tilde K_t
\end{align*}
between successive mode jumps. If $i=j$, then for every $k\ne i$ introduce three clocks with rates
\begin{align}
 m_{ik}(x,y)&=\min\{\lambda_{ik}(x),\lambda_{ik}(y)\},\notag\\
 p_{ik}(x,y)&=[\lambda_{ik}(x)-\lambda_{ik}(y)]^+,\notag\\
 \tilde p_{ik}(x,y)&=[\lambda_{ik}(y)-\lambda_{ik}(x)]^+.
 \label{eq:coupled_clock_rates}
\end{align}
They produce, respectively, $(i,i)\to(k,k)$, $(i,i)\to(k,i)$, and $(i,i)\to(i,k)$. Since
\begin{equation*}
 m_{ik}+p_{ik}=\lambda_{ik}(x),\qquad
 m_{ik}+\tilde p_{ik}=\lambda_{ik}(y),
\end{equation*}
the first and second marginal clocks have exactly their prescribed rates. If $i\ne j$, use the transitions
\begin{align}
 (i,j)&\to(k,j) &&\text{at rate }\lambda_{ik}(x),\quad k\ne i,
 \label{eq:cross_first_clock}\\
 (i,j)&\to(i,k) &&\text{at rate }\lambda_{jk}(y),\quad k\ne j.
 \label{eq:cross_second_clock}
\end{align}
Thus the unequal-mode clocks are conditionally independent to first order in time.  If one of these jumps makes the modes equal, the rule \eqref{eq:coupled_clock_rates} is used thereafter. Boundedness of all rates permits the standard interlacing construction and rules out an accumulation of jumps \cite{nguyen2017properties}. The resulting transition kernel $Q_t$ is measurable and Markovian, and the preceding rate identities show that its marginals are both $P_t$. Let $\cL^{\rm c}$ denote its interior generator acting on functions that are $C^2$ in $(x,y)$ for each fixed mode pair. Under the synchronous coupling the diffusion part is $\nu\sum_{\ell=1}^n(\partial_{x_\ell}+\partial_{y_\ell})^2$, which vanishes on functions depending on the continuous variables only through $x-y$. Reflection is accounted for by separate finite-variation terms.

\emph{Step 2: continuous separation and reflection.}
Write
\begin{align*}
 \Delta_t&=X_t-\tilde X_t, & R_t&=\|\Delta_t\|,\\
 D_t&=D^{\sigma_t\tilde\sigma_t}(X_t,\tilde X_t)
 =a_{\sigma_t\tilde\sigma_t}R_t+\beta_{\sigma_t\tilde\sigma_t}.&&
\end{align*}
Between jumps the common Brownian terms cancel, so $\Delta_t$ has the continuous finite-variation differential
\begin{align}
 \dd\Delta_t={}&[f_i(X_t)-f_j(\tilde X_t)]\dd t
 -\boldsymbol n(X_t)\dd K_t
 +\boldsymbol n(\tilde X_t)\dd\tilde K_t.
 \label{eq:delta_differential}
\end{align}
For $R_t>0$, the chain rule applied to \eqref{eq:delta_differential} gives
\begin{align}
 \dd R_t={}&\frac{\langle\Delta_t,
 f_i(X_t)-f_j(\tilde X_t)\rangle}{R_t}\dd t\notag\\
 &-\frac{\langle\Delta_t,\boldsymbol n(X_t)\rangle}{R_t}\dd K_t
 +\frac{\langle\Delta_t,\boldsymbol n(\tilde X_t)\rangle}{R_t}
 \dd\tilde K_t.
 \label{eq:distance_chain_rule}
\end{align}
Lemma~\ref{lem:reflection} makes the two reflection terms nonpositive. The first term is the contribution of the drift fields to the change in spatial separation. Formula~\eqref{eq:distance_chain_rule} is used only for $R_t>0$; at $R_t=0$ we use the smooth approximation in Step~5 rather than differentiating $R_t$ directly. Synchronous Brownian motion cancels in the difference and introduces no quadratic-variation correction to this smooth approximation.

\emph{Step 3: equal-mode generator bound.}
Let $i=j$. Then $D_t=v_iR_t$. From \eqref{eq:distance_chain_rule}, Assumption \ref{ass:drift}, and the reflection sign,
\begin{align}
 \cL_{\rm cont}^{\rm c}D_t
 &\le \frac{v_i}{R_t}\langle\Delta_t,
 f_i(X_t)-f_i(\tilde X_t)\rangle\notag\\
 &\le-c_iv_iR_t.
 \label{eq:proof_same_cont}
\end{align}
For a simultaneous jump to $k$, the positions do not reset, and the cost changes from $v_iR_t$ to $v_kR_t$. Its exact generator contribution is
\begin{equation}
 m_{ik}(X_t,\tilde X_t)(v_k-v_i)R_t.
 \label{eq:simultaneous_exact}
\end{equation}
If $v_k-v_i\ge0$, use $m_{ik}\le\overline\lambda_{ik}$ in \eqref{eq:simultaneous_exact}.  If $v_k-v_i<0$, use $m_{ik}\ge\underline\lambda_{ik}$, the inequality reverses when the negative increment is multiplied. Both cases are summarized by
\begin{align}
 m_{ik}(X_t,\tilde X_t)(v_k-v_i)R_t
 \le{}&(v_k-v_i)^+\overline\lambda_{ik}R_t\notag\\
 &+(v_k-v_i)^-\underline\lambda_{ik}R_t.
 \label{eq:proof_simultaneous}
\end{align}
If only the first component jumps, the exact increment is
\begin{align*}
 &c_{\bm v,B}((X_t,k),(\tilde X_t,i))-v_iR_t\\
 &\qquad=\beta_{ki}+(a_{ki}-v_i)R_t\le\beta_{ik},
\end{align*}
where symmetry of $B$ and $a_{ki}=\min\{v_k,v_i\}\le v_i$ were used. The second-only increment obeys the same bound. Moreover,
\begin{align*}
 p_{ik}+\tilde p_{ik}
 &=|\lambda_{ik}(X_t)-\lambda_{ik}(\tilde X_t)|\\
 &\le L_{ik}R_t.
\end{align*}
The two unmatched clocks therefore contribute at most $L_{ik}\beta_{ik}R_t$. Summing the continuous, simultaneous, and unmatched terms over $k\ne i$ gives
\begin{equation}
 \cL^{\rm c}D_t\le S_i(\bm v,B)R_t.
 \label{eq:same_generator_complete}
\end{equation}
Condition \eqref{eq:certificate_same} and $D_t=v_iR_t$ now yield
\begin{equation}
 \cL^{\rm c}D_t\le-\eta v_iR_t=-\eta D_t.
 \label{eq:proof_same_final}
\end{equation}
At $R_t=0$, the two continuous states and their equal-mode drifts coincide; both unmatched jump-clock rates are zero, and a simultaneous jump leaves the original cost equal to zero. The generator bound for the smooth approximation, including this case, is established in Step~5.

\emph{Step 4: unequal-mode generator bound.}
Let $i\ne j$. The current cost is $D_t=a_{ij}R_t+\beta_{ij}$. The constant $\beta_{ij}$ is unaffected by the continuous motion, and \eqref{eq:distance_chain_rule} together with Assumption \ref{ass:cross_drift} gives
\begin{equation}
 \cL_{\rm cont}^{\rm c}D_t
 \le a_{ij}(\kappa_{ij}R_t+h_{ij}).
 \label{eq:proof_cross_cont}
\end{equation}
Under \eqref{eq:cross_first_clock}, a first-component jump $i\to k$ changes the cost by the exact amount
\begin{align}
 &c_{\bm v,B}((X_t,k),(\tilde X_t,j))-D_t\notag\\
 &\qquad=(a_{kj}-a_{ij})R_t+\beta_{kj}-\beta_{ij}.
 \label{eq:first_cross_increment}
\end{align}
Under \eqref{eq:cross_second_clock}, a second-component jump $j\to k$ changes it by
\begin{align}
 &c_{\bm v,B}((X_t,i),(\tilde X_t,k))-D_t\notag\\
 &\qquad=(a_{ik}-a_{ij})R_t+\beta_{ik}-\beta_{ij}.
 \label{eq:second_cross_increment}
\end{align}
Multiplying \eqref{eq:first_cross_increment} and \eqref{eq:second_cross_increment} by their respective rates and summing gives the two sums in \eqref{eq:cross_margin}. Adding
\eqref{eq:proof_cross_cont} therefore yields the generator bound
\begin{equation}
 \cL^{\rm c}D_t\le C_{ij}(X_t,\tilde X_t;\bm v,B).
 \label{eq:cross_generator_complete}
\end{equation}
Condition \eqref{eq:certificate_cross} yields
\begin{equation}
 \cL^{\rm c}D_t\le-\eta(a_{ij}R_t+\beta_{ij})=-\eta D_t.
 \label{eq:proof_cross_final}
\end{equation}
At $R_t=0$, the drift estimate in \eqref{eq:proof_cross_cont} is interpreted through the smooth approximation. The jump increments remain algebraically valid at coincident positions. Step~5 establishes the full inequality without differentiating the original norm at zero.

\emph{Step 5: regularization, localization, and passage to laws.}
The preceding calculations identify the required drift and jump bounds for the original cost $D^{ij}$ at $R>0$. They cannot by themselves justify applying It\^o's formula to $D^{ij}$ at $x=y$, where $R=\|x-y\|$ is not differentiable. We therefore work with a smooth auxiliary cost and derive its generator inequality directly.

For every $\eps>0$, put
\begin{align}
 R_\eps(x,y)&=\sqrt{\|x-y\|^2+\eps^2},\notag\\
 D_\eps^{ij}(x,y)&=a_{ij}R_\eps(x,y)+\beta_{ij}.
 \label{eq:regularized_cost}
\end{align}
We call $D_\eps^{ij}$ the \emph{regularized mode-pair cost}. Unlike $D^{ij}$, it is $C^2$ in $(x,y)$ for every $\eps>0$ and every fixed mode pair. Moreover,
\begin{equation}
 0\le R_\eps-R\le\eps,
 \qquad 0\le D_\eps^{ij}-D^{ij}\le a_{ij}\eps.
 \label{eq:cost_approximation_error}
\end{equation}
The drift contribution to the interior generator of the regularized cost is, with $\Delta=x-y$,
\begin{align}
 &\big\langle\nabla_x D_\eps^{ij},f_i(x)\big\rangle
 +\big\langle\nabla_y D_\eps^{ij},f_j(y)\big\rangle\notag\\
 &\qquad=\frac{a_{ij}}{R_\eps}
 \big\langle\Delta,f_i(x)-f_j(y)\big\rangle.
 \label{eq:regularized_drift_general}
\end{align}
This is the drift contribution to the \emph{cost}; the drift fields themselves are not regularized. The synchronous Brownian diffusion term of $\cL^{\rm c}$ vanishes on $D_\eps^{ij}$ because it depends on $x,y$ only through $x-y$.

\emph{Equal modes $(i=j)$.}
Here $D_\eps^{ii}=v_iR_\eps$. By Assumption~\ref{ass:drift}, the drift contribution satisfies
\begin{align}
 \frac{v_i}{R_\eps}\langle\Delta,f_i(x)-f_i(y)\rangle
 &\le-\frac{c_iv_iR^2}{R_\eps}\notag\\
 &=-c_iv_iR_\eps+\frac{c_iv_i\eps^2}{R_\eps}\notag\\
 &\le-c_iv_iR_\eps+v_ic_i^+\eps,
 \label{eq:regularized_equal_drift}
\end{align}
where $c_i^+=\max\{c_i,0\}$ and $\eps^2/R_\eps\le\eps$. This inequality holds for either sign of $c_i$.

A simultaneous $i\to k$ jump has cost increment $(v_k-v_i)R_\eps$. The estimate in \eqref{eq:proof_simultaneous} therefore remains valid with $R_t$ replaced by $R_\eps$. A jump of only one component has increment
\begin{equation*}
 \beta_{ik}+(a_{ik}-v_i)R_\eps\le\beta_{ik}.
\end{equation*}
The combined unmatched-clock rate is at most
$|\lambda_{ik}(x)-\lambda_{ik}(y)|\le L_{ik}R\le L_{ik}R_\eps$.
Consequently, for the interior generator,
\begin{align}
 \cL^{\rm c}D_\eps^{ii}
 &\le S_i(\bm v,B)R_\eps+v_ic_i^+\eps\notag\\
 &\le-\eta D_\eps^{ii}+\chi_{ii}\eps,
 \qquad \chi_{ii}=v_ic_i^+.
 \label{eq:regularized_equal_generator}
\end{align}
In contrast to the original cost, $D_\eps^{ii}(x,x)=v_i\eps$ need not vanish at coincident states. The residual term in \eqref{eq:regularized_equal_generator} is therefore essential and cannot be omitted.

\emph{Unequal modes $(i\ne j)$.}
Assumption~\ref{ass:cross_drift} and \eqref{eq:regularized_drift_general} bound the drift contribution by
\begin{equation*}
 a_{ij}\frac{R}{R_\eps}(\kappa_{ij}R+h_{ij}).
\end{equation*}
Here $h_{ij}\ge0$, and
\begin{equation}
 0\le R\Big(1-\frac{R}{R_\eps}\Big)
 =\frac{R}{R_\eps}(R_\eps-R)\le\eps.
 \label{eq:regularization_distance_bound}
\end{equation}
Writing the drift bound as
\begin{align*}
 &a_{ij}\frac{R}{R_\eps}(\kappa_{ij}R+h_{ij})\\
 &\quad=a_{ij}(\kappa_{ij}R+h_{ij})
 -a_{ij}\kappa_{ij}R\Big(1-\frac{R}{R_\eps}\Big)\\
 &\qquad-a_{ij}h_{ij}\Big(1-\frac{R}{R_\eps}\Big),
\end{align*}
we obtain
\begin{align}
 a_{ij}\frac{R}{R_\eps}(\kappa_{ij}R+h_{ij})
 \le a_{ij}(\kappa_{ij}R+h_{ij})
 +a_{ij}(-\kappa_{ij})^+\eps.
 \label{eq:regularized_unequal_drift}
\end{align}
The positive-part term is necessary because $\kappa_{ij}$ may be negative.

For a first-component jump $i\to k$, the regularized cost increment is
\begin{equation}
 (a_{kj}-a_{ij})R_\eps+\beta_{kj}-\beta_{ij},
 \label{eq:regularized_first_jump}
\end{equation}
which differs from \eqref{eq:first_cross_increment} by at most $|a_{kj}-a_{ij}|\eps$. A second-component jump $j\to k$ changes the regularized cost by
\begin{equation}
 (a_{ik}-a_{ij})R_\eps+\beta_{ik}-\beta_{ij},
 \label{eq:regularized_second_jump}
\end{equation}
and its difference from \eqref{eq:second_cross_increment} is at most $|a_{ik}-a_{ij}|\eps$. Multiplying these errors by their bounded rates yields
\begin{align}
 \cL^{\rm c}D_\eps^{ij}
 \le C_{ij}(x,y;\bm v,B)
 +\eps\bigg[&a_{ij}(-\kappa_{ij})^+\notag\\
 &+\sum_{k\ne i}\overline\lambda_{ik}|a_{kj}-a_{ij}|\notag\\
 &+\sum_{k\ne j}\overline\lambda_{jk}|a_{ik}-a_{ij}|\bigg].
 \label{eq:regularized_unequal_preliminary}
\end{align}
Condition~\eqref{eq:certificate_cross} and \eqref{eq:cost_approximation_error} give
\begin{align*}
 -\eta(a_{ij}R+\beta_{ij})
 &=-\eta D_\eps^{ij}+\eta a_{ij}(R_\eps-R)\\
 &\le-\eta D_\eps^{ij}+\eta a_{ij}\eps.
\end{align*}
Thus
\begin{equation}
 \cL^{\rm c}D_\eps^{ij}
 \le-\eta D_\eps^{ij}+\chi_{ij}\eps,
 \label{eq:regularized_unequal_generator}
\end{equation}
where the explicit constant is
\begin{align}
 \chi_{ij}={}&a_{ij}[(-\kappa_{ij})^++\eta]\notag\\
 &+\sum_{k\ne i}\overline\lambda_{ik}|a_{kj}-a_{ij}|\notag\\
 &+\sum_{k\ne j}\overline\lambda_{jk}|a_{ik}-a_{ij}|.
 \label{eq:regularization_error_constant}
\end{align}
Let $\chi=\max_{i,j\in\M}\chi_{ij}<\infty$, where $\chi_{ii}$ is given in \eqref{eq:regularized_equal_generator}.

\emph{Reflection and It\^o-Dynkin formula.}
For the current mode pair $(i,j)$, the contribution of reflection to the regularized cost along the coupled process is
\begin{align}
 \dd D_{\eps,t}\big|_{\rm ref}
 ={}&-\frac{a_{ij}}{R_\eps}
 \langle X_t-\tilde X_t,\boldsymbol n(X_t)\rangle\dd K_t\notag\\
 &+\frac{a_{ij}}{R_\eps}
 \langle X_t-\tilde X_t,\boldsymbol n(\tilde X_t)\rangle
 \dd\tilde K_t\le0,
 \label{eq:regularized_reflection_sign}
\end{align}
by convexity of $\Omega$ and $a_{ij}>0$. Both regulators vanish when $\Omega=\R^n$. In particular, reflection adds no positive error to the bounds \eqref{eq:regularized_equal_generator} and \eqref{eq:regularized_unequal_generator}.

On $\R^n$, let
\begin{equation*}
 \tau_N=\inf\{s\ge0:\|X_s\|+\|\tilde X_s\|\ge N\}\wedge N,
\end{equation*}
whereas on compact $\Omega$ put $\tau_N=N$. Along the coupled process define
\begin{equation*}
 D_{\eps,t}=D_\eps^{\sigma_t\tilde\sigma_t}(X_t,\tilde X_t).
\end{equation*}
For every fixed $\eps>0$, this cost is $C^2$ in the continuous state variables. Applying It\^o's formula to the switching semimartingale stopped at $t\wedge\tau_N$ and retaining the reflection finite-variation terms gives
\begin{align}
 &\mathbb E_{z,\tilde z}
 \left[e^{\eta(t\wedge\tau_N)}D_{\eps,t\wedge\tau_N}\right]\notag\\
 &\qquad\le D_{\eps,0}
 +\chi\eps\,\mathbb E_{z,\tilde z}
 \left[\int_0^{t\wedge\tau_N}e^{\eta s}\dd s\right]\notag\\
 &\qquad\le D_{\eps,0}
 +\chi\eps\frac{e^{\eta t}-1}{\eta}.
 \label{eq:stopped_regularized_bound}
\end{align}
For every mode pair, $a_{ij}\le\max_i v_i$, so
\begin{equation*}
 0\le D_{\eps,t}-D_t\le(\max_i v_i)\eps.
\end{equation*}
At fixed $N,t$, we may let $\eps\downarrow0$ in \eqref{eq:stopped_regularized_bound} to obtain
\begin{equation}
 \mathbb E_{z,\tilde z}
 \left[e^{\eta(t\wedge\tau_N)}D_{t\wedge\tau_N}\right]
 \le D_0.
 \label{eq:stopped_supermartingale}
\end{equation}
The at-most-linear drift growth and bounded jump rates imply nonexplosion, hence $\tau_N\uparrow\infty$ almost surely. Fatou's lemma now yields
\begin{equation}
 \mathbb E_{z,\tilde z}[D_t]\le e^{-\eta t}D_0.
 \label{eq:pointwise_final_bound}
\end{equation}
This is the pointwise estimate \eqref{eq:point_coupling} for the coupling kernel $Q_t$. Lemma~\ref{lem:coupling_reduction} proves \eqref{eq:semigroup_contraction}, and Proposition~\ref{prop:semigroup_fpk} extends the bound to weak measure solutions of \eqref{eq:weak_fpk}.
\end{proof}

\begin{corollary}[Invariant law]
\label{cor:invariant}
Suppose the assumptions and certificate conditions of Theorem~\ref{thm:contraction} hold. When $\Omega$ is compact, the semigroup admits a unique invariant probability measure $\mu_\star\in\Pone(\X)$ and
\begin{equation}
 \cW(\mu_0P_t,\mu_\star)
 \le e^{-\eta t}\cW(\mu_0,\mu_\star),
 \qquad\mu_0\in\Pone(\X).
 \label{eq:invariant_contraction}
\end{equation}
When $\Omega=\R^n$, the existence, uniqueness, and contraction estimate~\eqref{eq:invariant_contraction} also hold provided there exist a finite-valued function $V:\X\to[1,\infty)$ in the extended generator domain and constants $a>0$, $b\ge0$ such that
\begin{align}
 V(x,i)&\ge1+\|x\|,\notag\\
 \cA V(x,i)&\le-aV(x,i)+b,\quad(x,i)\in\X.
 \label{eq:foster_lyapunov}
\end{align}
Here $\cA V$ is the extended generator: $V(Z_t)-V(Z_0)-\int_0^t\cA V(Z_s)\dd s$ is a local martingale for $Z_t=(X_t,\sigma_t)$.
\end{corollary}

\begin{proof}
Fix an arbitrary initial state $z_0\in\X$ and define the Krylov-Bogoliubov occupation measures
\begin{equation}
 \bar\mu_T=\frac1T\int_0^T\delta_{z_0}P_s\dd s,
 \qquad T\ge1.
 \label{eq:occupation_measures}
\end{equation}
It is the family $\{\bar\mu_T:T\ge1\}$ whose tightness is required for the existence argument. When $\Omega$ is compact, $\X=\Omega\times\M$ is compact because $\M$ is finite, so this family is tight and every probability measure on $\X$ has a finite first moment.

Now suppose $\Omega=\R^n$. By the extended-generator condition, the process
\begin{equation*}
 M_t^V=V(Z_t)-V(z_0)-\int_0^t\cA V(Z_s)\dd s
\end{equation*}
is a local martingale. Choose a localization under which the stopped local martingale is a true martingale. Applying the localized Dynkin formula and using \eqref{eq:foster_lyapunov}, we obtain
\begin{align}
 &\mathbb E_{z_0}[V(Z_{t\wedge\tau_N})]
 +a\mathbb E_{z_0}
 \left[\int_0^{t\wedge\tau_N}V(Z_s)\dd s\right]\notag\\
 &\qquad\le V(z_0)+b\mathbb E_{z_0}[t\wedge\tau_N]
 \le V(z_0)+bt.
 \label{eq:lyapunov_localized}
\end{align}
The two terms on the left are nonnegative since $V\ge1$. Nonexplosion, Fatou's lemma, and Tonelli's theorem therefore yield
\begin{equation}
 \mathbb E_{z_0}[V(Z_t)]
 +a\int_0^t\mathbb E_{z_0}[V(Z_s)]\dd s
 \le V(z_0)+bt.
 \label{eq:lyapunov_integrated}
\end{equation}
Dropping the first term, dividing by $aT$, and using \eqref{eq:occupation_measures} gives
\begin{equation}
 \bar\mu_T(V)
 \le\frac{V(z_0)}{aT}+\frac ba
 \le\frac{V(z_0)+b}{a},\qquad T\ge1.
 \label{eq:occupation_moment}
\end{equation}
This is a uniform moment bound for the specified occupation measures, not a claim about every probability measure associated with the process.

For $R>0$, let $K_R=\{(x,i)\in\X:\|x\|\le R\}$. Since $\M$ is finite, $K_R$ is compact. By $V(x,i)\ge1+\|x\|$ and Markov's inequality,
\begin{align}
 \bar\mu_T(K_R^c)
 &\le\frac1R\int_\X\|x\|\,\bar\mu_T(\dd(x,i))\notag\\
 &\le\frac{\bar\mu_T(V)}R
 \le\frac{V(z_0)+b}{aR},\qquad T\ge1.
 \label{eq:occupation_tightness}
\end{align}
The bound tends to zero as $R\to\infty$ uniformly in $T\ge1$, proving tightness. Prokhorov's theorem yields a sequence $T_k\to\infty$ and a probability measure $\mu_\star$ such that $\bar\mu_{T_k}\Rightarrow\mu_\star$.

To verify invariance, fix $t\ge0$ and $g\in C_0(\X)$. By the semigroup property,
\begin{align}
 \bar\mu_T(P_tg)-\bar\mu_T(g)
 =\frac1T\bigg[&\int_T^{T+t}P_sg(z_0)\dd s\notag\\
 &-\int_0^tP_sg(z_0)\dd s\bigg],
 \label{eq:occupation_invariance_identity}
\end{align}
whence
\begin{equation}
 |\bar\mu_T(P_tg)-\bar\mu_T(g)|
 \le\frac{2t}{T}\|g\|_\infty\longrightarrow0.
 \label{eq:occupation_invariance_bound}
\end{equation}
The Feller property gives $P_tg\in C_0(\X)$, allowing passage to the weak limit along $T_k$. Thus $\mu_\star(P_tg)=\mu_\star(g)$ for every $g\in C_0(\X)$, or equivalently $\mu_\star P_t=\mu_\star$.

On $\R^n$, it remains to verify that this invariant measure belongs to $\Pone(\X)$, the class to which Theorem~\ref{thm:contraction} applies. The function $(x,i)\mapsto\|x\|$ is nonnegative and continuous. By the Portmanteau theorem and \eqref{eq:occupation_moment},
\begin{align}
 \int_\X\|x\|\,\mu_\star(\dd(x,i))
 &\le\liminf_{k\to\infty}
 \int_\X\|x\|\,\bar\mu_{T_k}(\dd(x,i))\notag\\
 &\le\frac{V(z_0)+b}{a}<\infty.
 \label{eq:invariant_first_moment}
\end{align}
In particular, we do not need to assume that $V$ is continuous: lower semicontinuity is applied to $\|x\|$, not to $V$. On compact $\Omega$, the finite first moment follows directly from boundedness of the spatial domain.

Theorem~\ref{thm:contraction} and invariance now give, for every $\mu_0\in\Pone(\X)$,
\begin{align*}
 \cW(\mu_0P_t,\mu_\star)
 &=\cW(\mu_0P_t,\mu_\star P_t)\\
 &\le e^{-\eta t}\cW(\mu_0,\mu_\star),
\end{align*}
which is \eqref{eq:invariant_contraction}.

Finally we prove uniqueness among \emph{all} invariant probability measures, without assuming that another invariant law has a finite first moment. Fix $z\in\X$. Since $\delta_z\in\Pone(\X)$, contraction yields
\begin{equation*}
 \cW(\delta_zP_t,\mu_\star)
 \le e^{-\eta t}\cW(\delta_z,\mu_\star)\longrightarrow0.
\end{equation*}
By \eqref{eq:transport_equivalence}, $\delta_zP_t\Rightarrow\mu_\star$ for each fixed $z$. Let $\widehat\mu$ be any invariant probability measure, with no moment condition imposed. For every bounded continuous function $g$, invariance gives
\begin{equation*}
 \widehat\mu(g)
 =\int_\X P_tg(z)\,\widehat\mu(\dd z).
\end{equation*}
Since $P_tg(z)\to\mu_\star(g)$ pointwise in $z$ and $|P_tg(z)|\le\|g\|_\infty$, dominated convergence yields $\widehat\mu(g)=\mu_\star(g)$. Consequently $\widehat\mu=\mu_\star$.
\end{proof}

\subsection{Consistency checks and special cases}

The certificate contains several familiar limits. If $M=1$, there are no jump or unequal-mode conditions and Theorem \ref{thm:contraction} reduces to the synchronous-coupling estimate
\begin{equation*}
 W_1(\mu P_t,\tilde\mu P_t)
 \le e^{-c_1t}W_1(\mu,\tilde\mu).
\end{equation*}
If the mode weights are all equal, $\cW$ is a genuine Wasserstein metric, but simultaneous mode jumps do not contribute to the equal-mode margin. In particular,
\begin{equation*}
 S_i=-c_iv+\sum_{k\ne i}L_{ik}\beta_{ik}.
\end{equation*}
Thus an unequal weight vector is not a cosmetic generalization: it is what allows frequent jumps toward lower-cost modes to offset a negative $c_i$.

For a state-independent switching matrix $Q=[q_{ik}]$, where $q_{ik}\ge0$ for $k\ne i$ and $q_{ii}=-\sum_{k\ne i}q_{ik}$, the rate-Lipschitz penalties vanish and the equal-mode conditions are exactly
\begin{equation}
 \big[Q\bm v-\operatorname{diag}(c_1,\ldots,c_M)\bm v\big]_i
 \le-\eta v_i,\qquad i\in\M.
 \label{eq:constant_rate_reduction}
\end{equation}
Equation \eqref{eq:constant_rate_reduction} is a positive-vector inequality for a Metzler matrix. State dependence changes this clean matrix test in two ways: unmatched clocks produce the $L_{ik}\beta_{ik}$ penalty, and unequal-mode pairs must satisfy the spatially resolved condition \eqref{eq:certificate_cross}.

\section{Certificate Computation and Conservatism}
\label{sec:computation}

\subsection{Weight synthesis}

Conditions \eqref{eq:certificate_same}-\eqref{eq:certificate_cross} are homogeneous in $(\bm v,B)$, so the weights cannot be unique without normalization. We use
\begin{equation}
 \sum_{i=1}^M v_i=1,\qquad v_i\ge\eps_v,\qquad
 \beta_{ij}\ge\eps_\beta\quad(i\ne j),
 \label{eq:normalization}
\end{equation}
with small prescribed $\eps_v,\eps_\beta>0$.

\begin{proposition}[Fixed-order linearity]
\label{prop:fixed_order_lp}
Fix a total order of the weights, a decay rate $\eta$, and a collection of state pairs at which \eqref{eq:certificate_cross} is imposed. Under \eqref{eq:normalization}, the graph constraints \eqref{eq:graph_metric} and the sampled certificate conditions form a linear feasibility problem in $(\bm v,B)$.
\end{proposition}

\begin{proof}
For every ordered pair $(i,k)$, the prescribed order fixes one of the two affine expressions
\begin{equation*}
 (v_k-v_i)^+\overline\lambda_{ik}
 +(v_k-v_i)^-\underline\lambda_{ik}
 =\begin{cases}
  (v_k-v_i)\overline\lambda_{ik},&v_k\ge v_i,\\
  (v_k-v_i)\underline\lambda_{ik},&v_k\le v_i.
 \end{cases}
\end{equation*}
The same order fixes $a_{ij}$ to either $v_i$ or $v_j$.  At a sampled pair $(x,y)$, the quantities $r$, $\lambda_{ik}(x)$, and $\lambda_{jk}(y)$ are data. Since $\eta$ is also fixed, $\eta v_i$ and $\eta a_{ij}$ are linear rather than bilinear terms. Thus \eqref{eq:certificate_same} and every sampled instance of \eqref{eq:certificate_cross} are affine in $(\bm v,B)$. The normalization, positivity bounds, symmetry, and triangle inequalities for $B$ are linear as well, completing the reduction.
\end{proof}

On a continuum, the problem is a semi-infinite linear program (LP). A finite grid alone gives only a sampled certificate. A certified separation oracle, based for example on interval bounds or deterministic global optimization, converts the procedure into a proof over $\Omega^2$. Algorithm \ref{alg:synthesis} states the resulting computation.

\begin{algorithm}[t]
\caption{Mode-weight and decay-rate synthesis}
\label{alg:synthesis}
\begin{algorithmic}[1]
\REQUIRE Bounds, domain $\Omega$, and prescribed tolerances
\FOR{each ordering of the components of $\bm v$}
 \STATE Initialize a finite LP constraint set $\mathcal T\subset\Omega^2$
 \STATE Select a candidate $\eta$ by bracketing and bisection
 \REPEAT
  \STATE Solve the LP at $\eta$ using $\mathcal T$
  \IF{the LP is infeasible}
   \STATE Reject the candidate rate; \textbf{break}
  \ENDIF
  \STATE Evaluate $F_{ij}$ on a $\delta$-net $\mathcal S\subset\Omega^2$
  \IF{a sampled $F_{ij}(\hat x,\hat y)>0$}
   \STATE Add the violating state pairs to $\mathcal T$
  \ELSIF{$\max_{\mathcal S}F_{ij}+H_{ij}\delta\le0$ for all $i\ne j$}
   \STATE Certify the candidate rate; \textbf{break}
  \ELSE
   \STATE Refine the mesh or record an inconclusive test
  \ENDIF
 \UNTIL{a certified decision or a stopping limit is reached}
 \STATE Continue bisection using only certified decisions
\ENDFOR
\STATE \textbf{return} the largest verified rate and certificate
\end{algorithmic}
\end{algorithm}

The separation step can itself be certified by a finite mesh. Let
$R_\Omega=\sup_{x,y\in\Omega}\|x-y\|$ and define
\begin{equation*}
 F_{ij}(x,y)=C_{ij}(x,y;\bm v,B)
 +\eta(a_{ij}\|x-y\|+\beta_{ij}).
\end{equation*}
For $k\ne i$, set
\begin{equation*}
 A^i_{k}=a_{kj}-a_{ij},\quad
 B^i_{k}=\beta_{kj}-\beta_{ij},
\end{equation*}
and set $A^j_k=a_{ik}-a_{ij}$ and $B^j_k=\beta_{ik}-\beta_{ij}$.

\begin{proposition}[Finite-mesh certificate]
\label{prop:mesh_certificate}
On compact $\Omega$, $F_{ij}$ is Lipschitz with respect to $\|(x,y)-(x',y')\|_\oplus=\|x-x'\|+\|y-y'\|$. One valid constant is
\begin{align}
 H_{ij}={}&a_{ij}(|\kappa_{ij}|+\eta)\notag\\
 &+\sum_{k\ne i}\!\left[
 L_{ik}(|A^i_k|R_\Omega+|B^i_k|)
 +\overline\lambda_{ik}|A^i_k|\right]\notag\\
 &+\sum_{k\ne j}\!\left[
 L_{jk}(|A^j_k|R_\Omega+|B^j_k|)
 +\overline\lambda_{jk}|A^j_k|\right].
 \label{eq:mesh_lipschitz}
\end{align}
If $\mathcal S\subset\Omega^2$ is a $\delta$-net in this product norm and
\begin{equation}
 F_{ij}(\hat x,\hat y)\le-H_{ij}\delta
 \quad\text{for every }(\hat x,\hat y)\in\mathcal S,
 \label{eq:mesh_test}
\end{equation}
then \eqref{eq:certificate_cross} holds on all of $\Omega^2$.
\end{proposition}

\begin{proof}
Take $z=(x,y)$, $z'=(x',y')$, and put $d_\oplus(z,z')=\|x-x'\|+\|y-y'\|$. The reverse triangle inequality gives
\begin{align}
 |\|x-y\|-\|x'-y'\||
 &\le\|(x-y)-(x'-y')\|\notag\\
 &\le d_\oplus(z,z').
 \label{eq:mesh_distance_step}
\end{align}
Consequently, the sum of the unequal-mode drift term and the decay term has Lipschitz constant at most $a_{ij}(|\kappa_{ij}|+\eta)$ and the constants $a_{ij}h_{ij}$ and $\eta\beta_{ij}$ make no contribution.

For a typical first-component jump term $G(x,y)=\lambda_{ik}(x)(A r+B)$, add and subtract $\lambda_{ik}(x')(Ar+B)$ to obtain
\begin{align*}
 |G(x,y)-G(x',y')|
 &\le |\lambda_{ik}(x)-\lambda_{ik}(x')|\,|Ar+B|\\
 &\quad+\lambda_{ik}(x')|A|\,|r-r'|\\
 &\le L_{ik}\|x-x'\|(|A|R_\Omega+|B|)\\
 &\quad+\overline\lambda_{ik}|A|d_\oplus(z,z')\\
 &\le\Big\{L_{ik}(|A|R_\Omega+|B|)\notag\\
 &\hspace{3.4em}+\overline\lambda_{ik}|A|\Big\}d_\oplus(z,z').
\end{align*}
The second-mode terms obey the same estimate with $j,k$ and $y$ in place of $i,k$ and $x$. Summing all constants proves \eqref{eq:mesh_lipschitz}.

For an arbitrary $z\in\Omega^2$, choose $\hat z\in\mathcal S$ with $d_\oplus(z,\hat z)\le\delta$. The Lipschitz estimate and \eqref{eq:mesh_test} then give, step by step,
\begin{equation*}
 F_{ij}(z)\le F_{ij}(\hat z)+H_{ij}d_\oplus(z,\hat z)
 \le-H_{ij}\delta+H_{ij}\delta=0.
\end{equation*}
This is exactly \eqref{eq:certificate_cross}.
\end{proof}

An unbuffered grid check is diagnostic only. Condition \eqref{eq:mesh_test} adds the discretization margin needed for a global certificate.

For small $M$, enumerating weight orders is practical. For larger mode sets, the order selection can be embedded in a mixed-integer linear formulation or handled by branch and bound. Even after \eqref{eq:normalization}, the maximizing certificate need not be unique. A reproducible tie break is to fix $\eta$ within the requested tolerance of its optimum, minimize $\sum_{i<j}\beta_{ij}$, and then minimize the spread of the $v_i$ using auxiliary variables.

\subsection{How the envelopes interact}

Define the feasibility residuals
\begin{equation}
 M_i=S_i+\eta v_i,\qquad
 M_{ij}=C_{ij}+\eta(a_{ij}r+\beta_{ij}).
 \label{eq:residuals}
\end{equation}
The certificate requires all residuals to be nonpositive.

\begin{proposition}[Envelope sensitivity]
\label{prop:sensitivity}
For fixed $(\bm v,B,\eta)$ and a fixed weight order, the following changes in the residuals are exact:
\begin{align}
 c_i+\Delta c_i&:\quad \Delta M_i=-v_i\Delta c_i,
 \label{eq:sens_ci}\\
 h_{ij}+\Delta h_{ij}&:\quad
 \Delta M_{ij}=a_{ij}\Delta h_{ij},
 \label{eq:sens_h}\\
 \kappa_{ij}+\Delta\kappa_{ij}&:\quad
 \Delta M_{ij}=a_{ij}r\Delta\kappa_{ij},
 \label{eq:sens_kappa}\\
 L_{ik}+\Delta L_{ik}&:\quad
 \Delta M_i=\beta_{ik}\Delta L_{ik}.
 \label{eq:sens_L}
\end{align}
\end{proposition}

\begin{proof}
Fixing the weight order keeps every selected branch in $(v_k-v_i)^\pm$ and every $a_{ij}$ unchanged. In
\eqref{eq:same_margin}, the coefficient $c_i$ occurs only in the term $-c_iv_i$, subtracting the original residual from the perturbed one therefore gives $-v_i\Delta c_i$. In \eqref{eq:cross_margin}, $h_{ij}$ and $\kappa_{ij}$ occur only through $a_{ij}(\kappa_{ij}r+h_{ij})$, so their increments are respectively $a_{ij}\Delta h_{ij}$ and $a_{ij}r\Delta\kappa_{ij}$. Finally, $L_{ik}$ occurs in \eqref{eq:same_margin} only as $L_{ik}\beta_{ik}$, which changes by $\beta_{ik}\Delta L_{ik}$. The decay terms in \eqref{eq:residuals} are fixed in all four perturbations, proving \eqref{eq:sens_ci}-\eqref{eq:sens_L}.
\end{proof}

Equations \eqref{eq:sens_ci}-\eqref{eq:sens_L} expose the main tradeoff. A larger graph cost $\beta_{ik}$ can create a stronger decrease when an unequal-mode jump removes a mode mismatch, but the same $\beta_{ik}$ amplifies the unmatched-clock penalty $L_{ik}\beta_{ik}$ in \eqref{eq:same_margin}. More negative $c_i$ makes an expansive mode harder to stabilize, while larger $h_{ij}$ penalizes even coincident spatial states and larger $\kappa_{ij}$ is most damaging at the diameter of the domain. Local partitions or state-dependent envelopes may reduce conservatism by replacing one global worst case with several smaller separation problems, at the cost of more constraints.

\section{Numerical Studies}
\label{sec:example}

\subsection{One-dimensional analytical certificate}

Consider $\Omega=[-1,1]$, normal reflection at the endpoints, and $\nu=1$. The drifts and switching rates are
\begin{align}
 f_1(x)&=0.5x,& f_2(x)&=-2x,\label{eq:example_drifts}\\
 \lambda_{12}(x)&=\gamma+0.1\arctan(|x|),&
 \lambda_{21}(x)&=\delta=0.05.\label{eq:example_rates}
\end{align}
Mode 1 is expansive with $c_1=-0.5$, whereas Mode 2 is contractive with $c_2=2$. The common noise and normal reflection give the no-flux FPK equation used in the numerical calculation. Moreover, $L_{12}\le0.1$, $L_{21}=0$, and the unequal-mode drift bound holds with
\begin{equation}
 \kappa_{12}=\kappa_{21}=0,\qquad h_{12}=h_{21}=2.5.
 \label{eq:example_envelopes}
\end{equation}

Set $v_1=1$, $v_2=q\in(0,1)$, and $\beta_{12}=\beta_{21}=b$. Homogeneity permits $v_1=1$ instead of \eqref{eq:normalization}. The two equal-mode inequalities reduce to
\begin{align}
 0.5+\gamma(q-1)+0.1b+\eta&\le0,\label{eq:example_same1}\\
 -2q+\delta(1-q)+\eta q&\le0.\label{eq:example_same2}
\end{align}
For either unequal-mode ordering, a certified upper bound is
\begin{equation}
 qh-b(\gamma+\delta)+\delta(1-q)r+\eta(qr+b)\le0,
 \quad r\in[0,2],
 \label{eq:example_cross}
\end{equation}
where $h=2.5$. The left side is affine and increasing in $r$, so the endpoint $r=2$ is sufficient.

These scalar inequalities follow without a numerical relaxation. In Mode 1, $q-1<0$, the lower bound of $\lambda_{12}$ is $\gamma$, and $L_{12}=0.1$, hence
\begin{equation*}
 S_1=0.5+\gamma(q-1)+0.1b.
\end{equation*}
In Mode 2, $1-q>0$ and the rate is the constant $\delta$, which gives $S_2=-2q+\delta(1-q)$. For an unequal-mode pair, $a_{12}=q$, a jump $1\to2$ removes the graph cost and contributes $-b\lambda_{12}(x)$, whereas a jump $2\to1$ contributes $\delta[(1-q)r-b]$. Since the first contribution is negative, its largest value is obtained at the lower rate bound $\lambda_{12}(x)=\gamma$.  Adding the drift envelope $qh$ and the decay term $\eta(qr+b)$ yields \eqref{eq:example_cross}. Finally,
$\delta(1-q)+\eta q>0$, so its left side increases with $r$ and the diameter endpoint $r=2$ is the exact worst case for this bound.

Table \ref{tab:certificate_audit} reports the three maximal residuals in \eqref{eq:example_same1}-\eqref{eq:example_cross}. Negative entries certify the claimed rate. When $\gamma=0.2$, \eqref{eq:example_same1} is larger than $0.3$ even as $q,b,\eta\downarrow0$, so this family of weights cannot certify a positive rate. This is a failure of the sufficient certificate, not a proof of instability.

\begin{table}[t]
\caption{Certificate residuals for the two-mode example ($b=1$)}
\label{tab:certificate_audit}
\centering
\footnotesize
\begin{tabular}{ccccc}
\toprule
$\gamma$ & $(q,\eta)$ & Mode 1 & Mode 2 & Unequal modes \\
\midrule
2.0 & $(0.4,0.40)$ & $-0.20$ & $-0.61$ & $-0.27$\\
1.0 & $(0.2,0.10)$ & $-0.10$ & $-0.34$ & $-0.33$\\
0.2 & any positive rate & \multicolumn{3}{c}{infeasible by \eqref{eq:example_same1}}\\
\bottomrule
\end{tabular}
\end{table}

The interaction among the envelopes can be read directly from \eqref{eq:example_cross}. For fixed $(q,b,\eta)$, the largest admissible unequal-mode mismatch is
\begin{equation}
 h_{\max}=\frac{b(\gamma+\delta)-2\delta(1-q)-\eta(2q+b)}{q}.
 \label{eq:hmax}
\end{equation}
The two certified cases have $h_{\max}=3.175$ and $4.15$, respectively, both above the actual value $2.5$. Increasing $b$ raises $h_{\max}$ through the synchronous unequal-mode jump, but it simultaneously worsens \eqref{eq:example_same1} through the rate-Lipschitz penalty $0.1b$. This numerical identity illustrates the theoretical tradeoff in Proposition \ref{prop:sensitivity}.

\begin{figure}[htb!]
\centering
\includegraphics[width=0.92\linewidth]{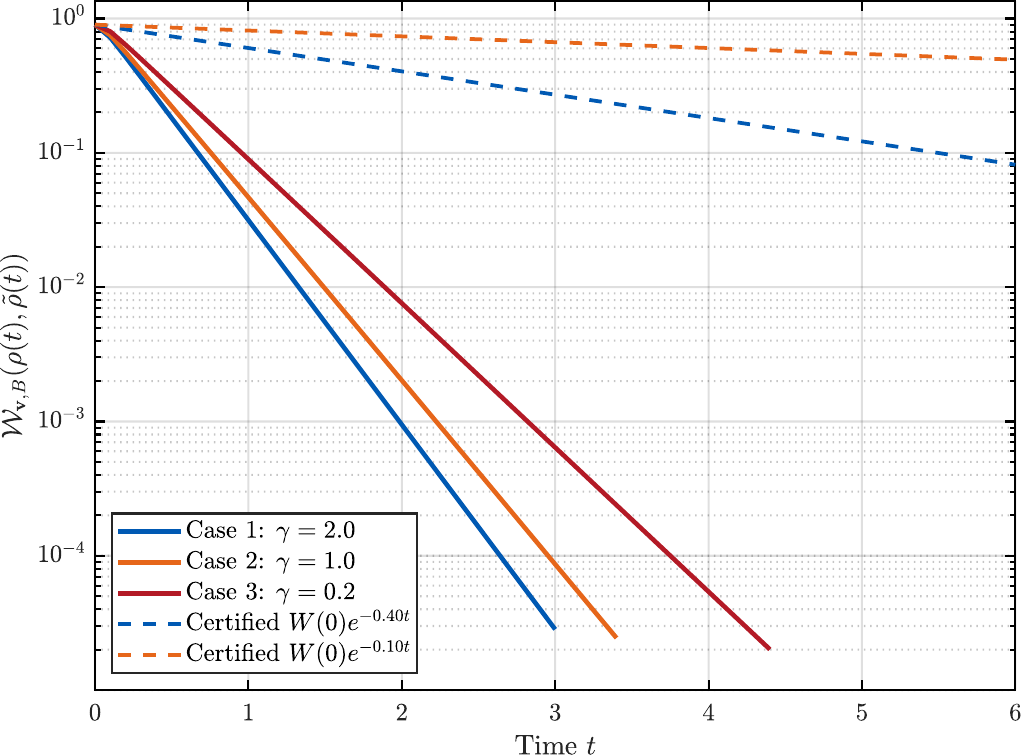}
\caption{Evolution of the mode-weighted transport discrepancy. The straight lines are the certified envelopes for $\gamma=2.0$ and $\gamma=1.0$. Each numerical curve is omitted after it first falls below the OT solver display threshold $2\times10^{-5}$.}
\label{fig:wasserstein_decay}
\end{figure}

\begin{figure}[htb!]
\centering
\includegraphics[width=0.92\linewidth]{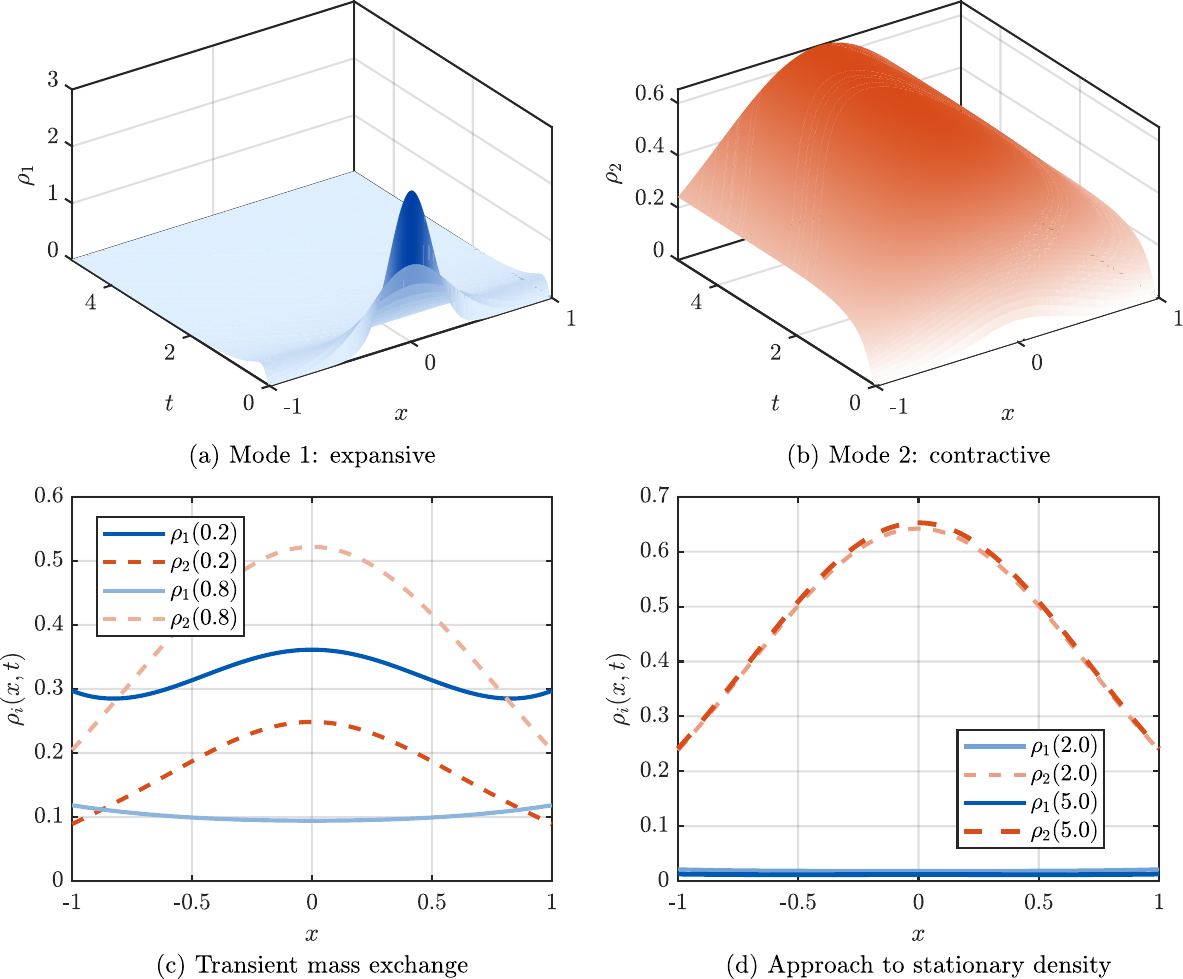}
\caption{No-flux FPK evolution on $[-1,1]$ for $\gamma=2.0$.}
\label{fig:fpk_evolution}
\end{figure}

\subsection{Distributional evolution and numerical verification}

For reproducibility, let
\begin{equation}
 g_m(x)=\frac{\exp[-(x-m)^2/(2\sigma_0^2)]}
 {\displaystyle\int_{-1}^{1}\exp[-(s-m)^2/(2\sigma_0^2)]\dd s},
 \qquad \sigma_0=0.15.
 \label{eq:initial_gaussian}
\end{equation}
The two laws compared in Fig. \ref{fig:wasserstein_decay} are $\rho^A_1(x,0)=g_{-0.45}(x)$ and $\rho^B_1(x,0)=g_{0.45}(x)$, with $\rho^A_2(x,0)=\rho^B_2(x,0)=0$. Denote the corresponding hybrid laws by $\rho(t)=(\rho_1^A(\cdot,t),\rho_2^A(\cdot,t))$ and $\widetilde\rho(t)=(\rho_1^B(\cdot,t),\rho_2^B(\cdot,t))$. These are the arguments of $\cW$ on the ordinate of Fig. \ref{fig:wasserstein_decay}. The density experiment in Fig. \ref{fig:fpk_evolution} instead uses $\rho_1(x,0)=g_0(x)$ and $\rho_2(x,0)=0$. Thus every initial law is normalized on the truncated computational domain, rather than on $\R$.

The no-flux FPK system is solved using the MATLAB \texttt{pdepe} solver, using relative and absolute tolerances $10^{-7}$ and $10^{-9}$. In the conservative form expected by that solver, the flux supplied for Mode $i$ is $\partial_x\rho_i-f_i\rho_i$, setting this quantity to zero at both endpoints implements \eqref{eq:no_flux}. The discrepancy calculation uses $81$ spatial nodes and $61$ output times on $[0,6]$, whereas the density plot uses $300$ nodes and $180$ output times on $[0,5]$. Trapezoidal masses are formed before each transport solve. Negative cell masses, if present, are replaced by zero, and the two marginals are then normalized separately before they enter the Kantorovich LP. In the reported runs no negative mass was removed and the largest normalization correction was $4.15\times10^{-6}$ and the largest postprocessed marginal-mass mismatch was $4.44\times10^{-16}$. The density arrays used for plotting were not clipped: their minimum was zero and their largest mass defect was $3.29\times10^{-8}$. Figure \ref{fig:wasserstein_decay} shows the mode-weighted discrepancy for the three switching intensities. The two certified trajectories are plotted with their exponential bounds. The case $\gamma=0.2$ is not certified by this weight family but nevertheless decreases in this particular simulation.

The no-flux density evolution for $\gamma=2.0$ is shown in Fig. \ref{fig:fpk_evolution}. The density spreads under the expansive first-mode drift while state-dependent switching transfers probability to Mode 2. Compactness gives an invariant law, and Theorem \ref{thm:contraction} makes it unique and globally attractive in $\cW$.

\subsection{Planar three-mode certificate synthesis}

We next apply Algorithm \ref{alg:synthesis} to a planar three-mode system. Let $\Omega=\{x\in\R^2:\|x\|\le0.5\}$ and $f_i(x)=A_ix$, where
\begin{equation}
\begin{split}
A_1&=\begin{bmatrix}0.45&-0.25\\0.25&0.30\end{bmatrix},\quad
A_2=\begin{bmatrix}-1.80&0.30\\-0.30&-1.10\end{bmatrix},\\
A_3&=\begin{bmatrix}-0.90&-0.35\\0.35&-1.70\end{bmatrix}.
\end{split}
\label{eq:three_mode_drifts}
\end{equation}
Mode 1 is expansive and the other two modes are contractive, with $(c_1,c_2,c_3)=(-0.45,1.10,0.90)$. All six off-diagonal rates vary with both state coordinates:
\begin{align}
\lambda_{12}(x)&=3.20+0.50x_1+0.35x_2,\notag\\
\lambda_{13}(x)&=2.40-0.35x_1+0.45x_2,\notag\\
\lambda_{21}(x)&=0.25+0.12x_1-0.08x_2,\notag\\
\lambda_{23}(x)&=0.90+0.18x_1+0.14x_2,\notag\\
\lambda_{31}(x)&=0.20-0.10x_1+0.09x_2,\notag\\
\lambda_{32}(x)&=0.80+0.14x_1-0.16x_2.
\label{eq:three_mode_rates}
\end{align}
Writing a rate as $\ell_{ij}+d_{ij}^{\top}x$ gives the exact disk bounds
\begin{align*}
 \underline\lambda_{ij}&=\ell_{ij}-0.5\|d_{ij}\|,&
 \overline\lambda_{ij}&=\ell_{ij}+0.5\|d_{ij}\|,\\
 L_{ij}&=\|d_{ij}\|.&&
\end{align*}
The smallest lower bound is $0.133$, so every rate remains positive. For the unequal-mode drift estimate we use $\kappa_{ij}=\lambda_{\max}[(A_i+A_i^\top)/2]=-c_i$ and $h_{ij}=0.5\|A_i-A_j\|_2$. These choices follow directly by writing $A_ix-A_jy=A_i(x-y)+(A_i-A_j)y$, and therefore do not rely on sampled drift values.

The state-dependent rate fields and the final separator diagnostics are collected in Fig. \ref{fig:three_mode_certificate}. The numerical construction of that certificate is described next.

\begin{figure*}[htb!]
\centering
\includegraphics[width=0.88\textwidth]{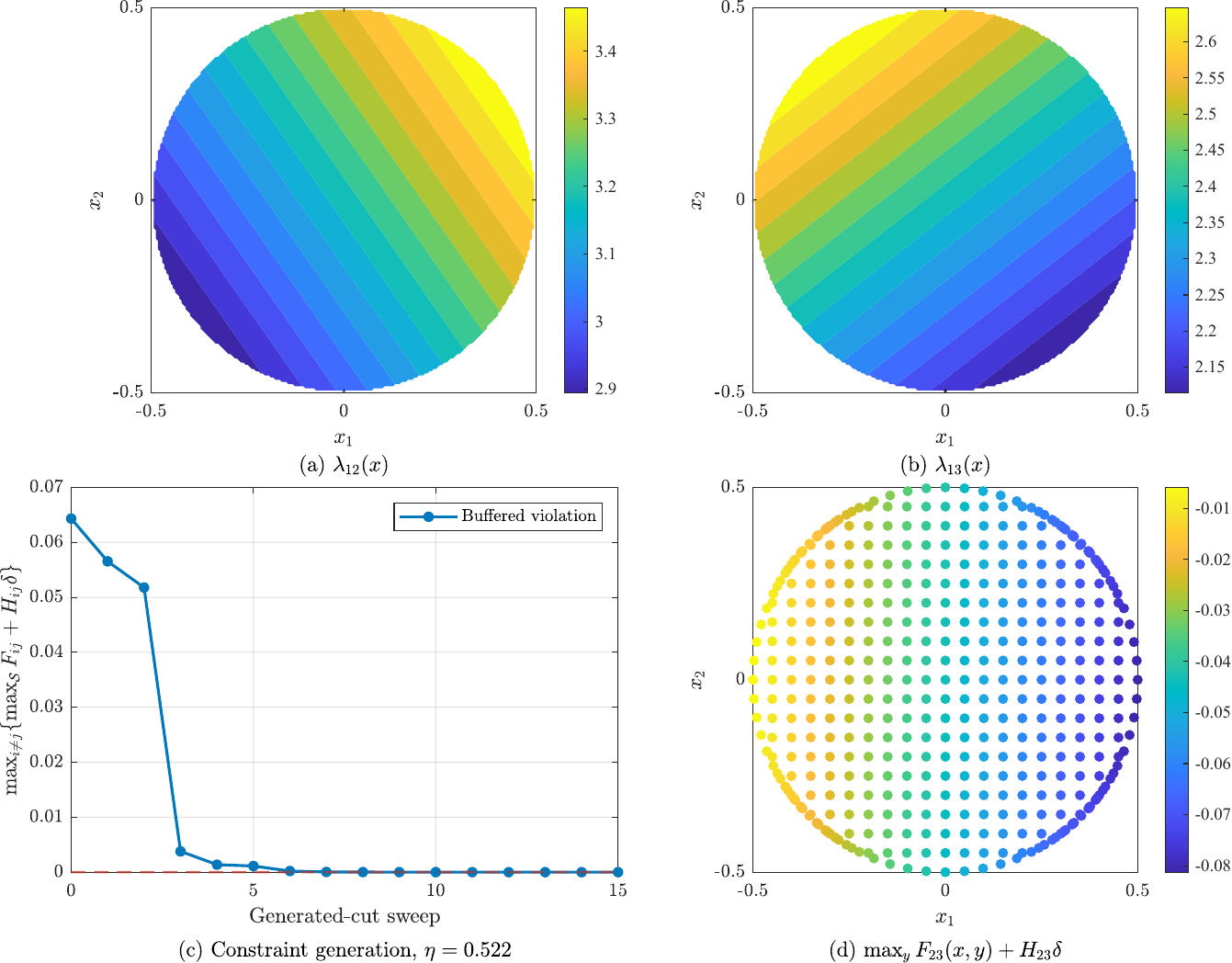}
\caption{Planar three-mode synthesis. Panels (a)-(b) show two of the state-dependent jump rates in \eqref{eq:three_mode_rates} and the remaining four also depend affinely on both spatial coordinates. Panel (c) shows the buffered separator residual during the final synthesis run at $\eta=0.522449$. Panel (d) gives $\max_y F_{23}(x,y)+H_{23}\delta$ for the retained rate $\eta=0.50$.}
\label{fig:three_mode_certificate}
\end{figure*}

The implementation enumerates all six weight orders, uses $\eps_v=0.04$, $\eps_\beta=0.02$, and bisects $\eta$ to a tolerance of $5\times10^{-4}$.  We also impose $\beta_{ij}\le4$ in the reported LPs; this bound is nonbinding at the selected certificate and makes the feasible search region bounded.  Starting from five state pairs for each ordered mode pair, the implementation scans the prescribed mesh and adds a maximizer whenever a sampled constraint is violated. A buffered residual that is positive without an actual sampled violation would require mesh refinement before certification.  The numerical bisection boundary estimates for the three orders that remain feasible within solver tolerance were
\begin{equation*}
\begin{array}{c|ccc}
\text{order}&v_3\ge v_1\ge v_2&v_1\ge v_2\ge v_3&v_1\ge v_3\ge v_2\\ \hline
\widehat\eta&0.0659&0.4707&0.5229\\
\text{new cuts}&10&14&15
\end{array}
\end{equation*}
and the other three orders were infeasible even at $\eta=0$. For the best order, the last synthesis run at $\eta=0.522449$ had maximum buffered residual $+4.825\times10^{-9}$, within solver tolerance, and returned
\begin{align}
\bm v&=(0.42958184,\ 0.27618608,\ 0.29423208),\notag\\
(\beta_{12},\beta_{13},\beta_{23})
&=(0.25459908,\ 0.26814957,\ 0.24155606).
\label{eq:three_mode_weights}
\end{align}
We therefore treat $\eta=0.522449$ only as a numerical boundary estimate and retain $\eta=0.50$ as the strict certificate. With the weights in \eqref{eq:three_mode_weights}, the three equal-mode residuals are $-1.40\times10^{-2}$, $-6.20\times10^{-3}$, and $-6.61\times10^{-3}$.

It remains to verify the continuum of unequal-mode inequalities. Projecting the $21\times21$ lattice of spacing $h=0.05$ on $[-0.5,0.5]^2$ onto the disk and retaining distinct projected nodes gives $433$ state points. Metric projection onto a closed convex set is nonexpansive.  Every point of the disk has a lattice point within $\sqrt{2}h/2$ before projection, and projection is nonexpansive. Thus the projected state mesh has covering radius at most $\sqrt{2}h/2$; its Cartesian square is a $\delta$-net of $\Omega^2$ with
\begin{equation*}
 \delta=\sqrt{2}h\approx0.070711
\end{equation*}
in the product norm of Proposition \ref{prop:mesh_certificate}. Each separator sweep evaluates all $433^2$ state pairs for each of the six ordered mode pairs. The final LP contains the $30$ initial unequal-mode constraints and $15$ generated cuts. Table \ref{tab:three_mode_mesh_audit} summarizes the resulting finite-mesh verification. In every row, $\max_{\mathcal S}F_{ij}+H_{ij}\delta<0$. Proposition \ref{prop:mesh_certificate} therefore certifies $F_{ij}(x,y)\le0$ for every $(x,y)\in\Omega^2$, not only at the computed nodes.

\begin{table}[htbp]
\caption{Finite-mesh verification for the planar three-mode example}
\label{tab:three_mode_mesh_audit}
\centering
\footnotesize
\begin{tabular}{cccc}
\toprule
$(i,j)$ & $H_{ij}\delta$ & $\max_{\mathcal S}F_{ij}$ & certified residual\\
\midrule
$(1,2)$ & $3.952\!\times\!10^{-2}$ & $-5.187\!\times\!10^{-2}$ & $-1.235\!\times\!10^{-2}$\\
$(1,3)$ & $4.500\!\times\!10^{-2}$ & $-5.809\!\times\!10^{-2}$ & $-1.309\!\times\!10^{-2}$\\
$(2,1)$ & $5.222\!\times\!10^{-2}$ & $-3.645\!\times\!10^{-1}$ & $-3.123\!\times\!10^{-1}$\\
$(2,3)$ & $4.135\!\times\!10^{-2}$ & $-4.721\!\times\!10^{-2}$ & $-5.861\!\times\!10^{-3}$\\
$(3,1)$ & $5.436\!\times\!10^{-2}$ & $-2.719\!\times\!10^{-1}$ & $-2.175\!\times\!10^{-1}$\\
$(3,2)$ & $3.744\!\times\!10^{-2}$ & $-4.721\!\times\!10^{-2}$ & $-9.767\!\times\!10^{-3}$\\
\bottomrule
\end{tabular}
\end{table}

Figure \ref{fig:three_mode_certificate} displays two rate fields and the constraint-generation record. Its last panel plots the buffered residual after maximization over the second spatial argument. The least favorable ordered pair is $(2,3)$, and its residual remains strictly below zero throughout the disk.

\section{Conclusion}
\label{sec:conclusion}

The mode-weighted cost turns state-dependent switching into an explicit pairwise contraction test. Equal-mode clock mismatches appear as rate-Lipschitz penalties and unequal-mode drift mismatch remains visible through a spatial term even before the discrete components meet. Normal reflection has a favorable sign on convex domains, which carries the semigroup estimate to the no-flux FPK system. Fixed weight order and fixed decay rate leave a semi-infinite linear feasibility problem. The planar three-mode example shows that order enumeration, bisection, constraint generation, and the explicit mesh buffer can be combined in a single synthesis procedure, yielding a full-domain certificate rather than a grid-only test. The shared isotropic diffusion coefficient permits exact cancellation under synchronous coupling and leads to the standard normal-reflection/no-flux correspondence. Mode-dependent or anisotropic diffusion would require a different spatial geometry or a covariance coupling. Global unequal-mode envelopes may also become conservative as the domain grows and partitioned envelopes offer one route to sharper conditions.

\bibliographystyle{IEEEtran}
\bibliography{ref}

\end{document}